\documentclass[11pt]{article}
\normalsize
\usepackage{amsfonts, amsthm, amsgen}
\usepackage{graphicx}
\usepackage{soul}
\usepackage[latin1]{inputenc}
\usepackage{tikz}
\usepackage{hyperref}
\usepackage{authblk}
\usepackage{amssymb,latexsym,amsmath}
\usepackage{amsthm}
\usepackage{algorithm}
\usepackage{algpseudocode}
\usepackage{caption}

\usepackage{multicol}
\algblock[Name]{Game}{EndGame}
\theoremstyle{definition}
\newtheorem{theorem}{Theorem}[section]

\newtheorem{lemma}[theorem]{Lemma}

\newtheorem{definition}[theorem]{Definition}

\newtheorem{agame}[theorem]{Game}

\newcommand{\Gal}{\mathrm{Gal}}

\newcommand{\F}{\mathbb F}

\newcommand{\N}{\mathbb N}
\newcommand{\Z}{\mathbb Z}

\newcommand{\Aut}{\mathrm{Aut}}

\title{\bf Public key cryptography based on skew dihedral group rings}

\author{Javier de la Cruz\\\vspace{-10pt}Universidad de Norte, Barranquilla\\ \vspace{5pt} Edgar Mart\' inez-Moro \\Universidad de Valladolid, Espa\~na\\ \vspace{5pt}Ricardo Villanueva-Polanco \\Universidad del Norte, Barranquilla, Colombia
}
\date{}

\begin{document}

\maketitle

\begin{abstract}
In this paper, we propose to use a skew dihedral group ring given by the group $D_{2n}$ and the finite field $\F_{q^2}$ for public-key cryptography. 
Using the ambient space $\F_{q^{2}}^{\theta} D_{2n}$ and a group homomorphism $\theta: D_{2n} \rightarrow \Aut(\F_{q^2})$, we introduce a key exchange protocol and present an analysis of its security. Moreover, we explore the properties of the resulting skew group ring $\F_{q^{2}}^{\theta} D_{2n}$, exploiting them to enhance our key exchange protocol. We also introduce a probabilistic public-key scheme derived from our key exchange protocol and obtain a key encapsulation mechanism (KEM) by applying a well-known generic transformation to our public-key scheme. Finally, we present a proof-of-concept implementation of our cryptographic constructions. To the best of our knowledge, this is the first paper that proposes a skew dihedral group ring for public-key cryptography.

\end{abstract}



\section{Introduction}

The availability of quantum computers in the forthcoming future will make current public-key schemes insecure. Therefore, there is a need for devising quantum-secure cryptographic public-key primitives as a replacement for the current public-key algorithms. This need undoubtedly has propelled research towards creating quantum-secure public-key schemes. There have been many proposed candidates so far, of which the most promising ones are classified into five families. These families are lattice-based cryptography, multivariate cryptography, hash-based cryptography, code-based cryptography, and supersingular elliptic curve isogeny cryptography. The third round of the post-quantum cryptography standardization process run by the National Institute of Standards and Technology (NIST) includes various candidates in each of the mentioned families~\cite{NIST}. 

However, recently a new promising family of cryptographic constructions, believed to be quantum-secure and based on variations of group rings~\cite{cryptosyme,CV}, has been introduced. In particular, the recent works~\cite{cryptosyme,CV} exploit the structure of dihedral twisted group rings to introduce cryptographic constructions. The work~\cite{cryptosyme} introduces a $2$-cocycle $\beta$ in order to construct a dihedral twisted group algebra $\F_q^{\beta}D_{2n}$. Over $\F_q^{\beta}D_{2n}$, the authors build a key-exchange protocol \`{a} la Diffie-Hellman and a probabilistic public-key scheme. Following an alternative approach, the authors of~\cite{CV} propose a key exchange protocol, a probabilistic public-key scheme, and a key encapsulation mechanism. They also introduce a $2$-cocycle $\alpha_{\lambda}$ to form the resulting twisted algebra $\F_q^{\alpha_{\lambda}}G$ non-equivalent to $\F_q^{\beta}D_{2n}$ for a non-square $\lambda$ in the field $\F_q$. They explore its properties and exploit them to enhance the introduced key exchange protocol. 

In other related works, the authors in \cite{CW-1} investigate right ideals as codes in twisted group rings. In particular, they characterize all linear codes that are twisted group codes in terms of their automorphism group.

Our work takes an alternative path by introducing what we call a skew dihedral group ring which is the main tool for constructing a key exchange protocol, a probabilistic public-key scheme, and a derived key encapsulation mechanism. We first formally define the notion of a skew group ring and explore some of its properties. We then study skew dihedral group rings, and later construct a specific skew dihedral group ring by defining the group homomorphism $\theta_\sigma: D_{2n} \rightarrow \Aut(\F_{q^2}) $ stated in Lemma \ref{homo}. In particular, given the presentation $G=D_{2n}=\langle x,y: x^n=y^2=1, yxy^{-1}=x^{-1}\rangle$ of the dihedral group, the map $\theta_\sigma(g)=\sigma$, where $\sigma(a)=a^q$ for all $a\in \F_{q^2}$, for $g=x^iy$, $i\in\{0, \ldots, n-1\}$ and $\theta_\sigma(g)=1$ otherwise is a group homomorphism. Over the resulting skew dihedral group ring $\F_{q^2}^{\theta} D_{2n}$, we realize our cryptographic constructions and analyze their security. Finally, we present a proof-of-concept implementation of our key encapsulation mechanism.

The outline of the paper is as follows. In Section~\ref{II} we show the basic definitions and results we need, whereas in Section~\ref{III} we show the concrete presentation of the dihedral group ring we use. Section~\ref{KEP} presents the proposed key exchange protocol and analyzes its intractability assumptions. Section~\ref{V} presents a probabilistic public-key encryption scheme, and Section~\ref{VI} introduces a key encapsulation mechanism using the ideas from the previous sections. Finally, Section~\ref{implementations} presents the pseudo-codes of a proof-of-concept Python implementation of our cryptographic constructions.

\section{Preliminaries} \label{II}

Let $\F_q$ be the finite field with $q=p^m$ elements where $p$ is a prime number and let $\Aut(\F_q)$ be the automorphism group of $\F_q$. Recall that any automorphism $\Theta\in\Aut(\F_q)$ of the finite field $\F_q$ is of the type $\Theta(x) = x^{p^j}$. Denote by $\Gal(\F_{q^k} , \F_q)$ the Galois group of $\F_{q^k}$ over $\F_q$, i.e. the set of all automorphisms of $\F_{q^k}$ that fix the subfield $\F_q$. It holds
that $\Aut(\F_{q^k}) = \Gal(\F_{q^k}, \F_q) \times \Aut( \F_q)$, in particular $\Aut(\F_q)= \Gal(\F_q, \F_p)$.

In the following paragraphs we summarise the definitions and properties we need on skew group rings.

\begin{definition} Let $G$ be a finite multiplicative group and let $\theta: G \rightarrow \Aut(\F_q)$ be a group homomorphism. The skew group ring $\F_q^\theta G$ is the set of all formal sums $\sum_{g\in G} a_g g$, where $a_g\in \F_q$, with the following skew multiplication
$$a_g g\cdot b_h h= a_g(\theta(g)(b_h)) gh.$$
\end{definition}

Note that as $\F_q$-vector space the skew group ring $\F_q^\theta G$ coincides with the group ring $\F_q G$. However, as rings not only may not coincide, but in general they are non-isomorphic. More precisely we have the following result.

\begin{lemma} Let $\theta$ and $\beta$ be homomorphisms of $G$ into $\Aut(\F_q)$. There is a $\F_q$-isomorphism $\F_q^\theta G \cong \F_q^{\beta}G$ mapping $a_g g$ to $a_g\delta(g)$ for some $\delta(g)\in G$ if and only $\delta: G \rightarrow G$ is an isomorphism such that $\theta(g)=\beta (\delta(g))$ for all $g\in G$.
\end{lemma}
\begin{proof} Denote by $a_g g \cdot b_h h$ the product in $\F_q^\theta G$ and denote $a_g g \ast b_h h$ the product in $\F_q^\beta G$. The image of $a_g g \cdot b_h h= a_g \theta (g)(b_h) gh $ in $\F_q^\theta G$ is  $a_g \theta(g)(b_h)\delta(gh)$. The product in $\F_q^\beta G$ of the images of $a_g g$ and $b_h h$ is $a_g \beta(\delta(g))(b_h)\delta(g)\delta(g)=a_g\beta(\delta(g))(b_h)\delta(gh).$ The two elements coincide if and only if $\theta(g)=\beta (\delta(g))$ for all $g\in G$.
\end{proof}

\begin{lemma}\label{anti} The map $\varphi:\F_q^{\theta}G \rightarrow \F_q^{{\theta}}G$, $ \sum_{g \in G} a_g g \mapsto \sum_{g \in G} \theta(g^{-1})(a_g) g^{-1}$, is an ring anti-isomorphism of $\F_q^{\theta}G$.
\end{lemma}
\begin{proof} 
Let $a_gg, b_hh \in \F_q^{\theta}G$. Then we have
$$\begin{array}{rl} \varphi( a_gg \cdot b_hh) &= \varphi(a_g \theta(g)(b_h) gh)=\theta (h^{-1} g^{-1})(a_g \theta(g)(b_h)) h^{-1} g^{-1}\\
& = \theta (h^{-1}) \theta (g^{-1})(a_g \theta(g) b_h) h^{-1}g^{-1}=(\theta(h^{-1}) \theta (g^{-1}) (a_g) \theta(h^{-1}) (b_h) h^{-1} g^{-1}\\
 &= \theta (h^{-1})(b_h) \theta(h^{-1}) \theta(g^{-1}) (a_g) h^{-1} g^{-1}=\theta (h^{-1}) (b_h) h^{-1} \cdot \theta( g^{-1}) (a_g) g^{-1} \\
 & = \varphi (b_h g) \varphi ( a_g g).
\end{array}$$
\end{proof}

\begin{definition} For an element $a = \sum_{g \in G}a_g g \in \F_q^\theta G$ we define its adjunct as
$$ \widehat{a}:=\varphi(a) = \sum_{g \in G}\theta(g^{-1})(a_g) g^{-1}.$$
\end{definition}

\section{A skew dihedral group ring} \label{III}
Let $G=D_{2n}=\langle x,y: x^n=y^2=1, yxy^{-1}=x^{-1}\rangle$ be a presentation of the dihedral group of order $2n$.


\begin{lemma}\label{lemmap}Let $C_n=\langle x\rangle$ be the cyclic subgroup of $D_{2n}$ generated by $x$. Then we have \begin{enumerate}
 \item $\F_q^{\theta}D_{2n}$ is a free $\F_q^{\theta}C_n$-module with basis $\{ 1,y\}$. Therefore $ \F_q^{\theta}D_{2n}=\F_q^{\theta}C_{n}\oplus \F_q^{\theta}C_{n}y$ as direct sum of $\F_q$-vector spaces.
 \item $\F_q^{\theta}C_{n} y \cong \F_q^{\theta}C_{n}$ as $\F_q^{\theta}C_{n}$-modules.
 \item For $a \in \F_q^{\theta}C_{n}y$, $ab \in \F_q^{\theta}C_{n} $ if $b \in \F_q^{\theta}C_{n}y$ or $ab \in \F_q^{\theta}C_{n}y$ if $b \in \F_q^{\theta}C_{n}$.

 \item If $a \in \F_q^{\theta}C_{n}$, then $\widehat{a} \in \F_q^{\theta}C_{n}$.
 \item If $a \in \F_q^{\theta}C_{n}y$, then $\widehat{a} \in \F_q^{\theta}C_{n}y$.
\end{enumerate}
\end{lemma}

\begin{proof} In what follows the symbol $[k]_n$ for $k\in \Z$ denotes $k \equiv [k]_n\mod n$.
\begin{enumerate}
 \item Since $\{1,y\}$ is a transversal for $C_n$ in $D_{2n}$, then $D_{2n}=C_n \cup C_ny$ and the assertion follows.
 \item Since $x^i\cdot x^j=x^{[i+j]_n}$ and $x^i\cdot x^{j}y=x^{[i+j]_n}y$ for all $i,j\in\{0, \ldots, n-1\}$, the assertions follow.
 \item Since $x^i y\cdot x^jy=x^{[i-j]_n}$ and $x^i y \cdot x^{j}=x^{[i-j]_n}y$ for all $i,j\in\{0, \ldots,n-1\}$, the assertions follow.

 \item Since $x^i\cdot x^{n-i}=1$ for all $i\in\{0, \ldots, n-1\}$, then the assertion follows.

 \item Since $(x^iy)^2=1$ for all $i\in\{0, \ldots, n-1\}$, then it follows.

\end{enumerate}
\end{proof}
%
%

\begin{definition} $ \quad$ 
\begin{enumerate} \item We define the $\theta$-reversible subspace of $\F_q^{\theta} C_{n}y$ as the vector subspace 
$$\Gamma_\theta=\{a=\sum_{i=0}^{n-1}a_i x^iy \in \F_q^{\theta} C_{n}y \mid a_i=a_{n-i}\; \textrm{for}\; i=1, \ldots, n-1 \}.$$

\item Given $a=\sum_{i=0}^{n-1}a_i x^iy \in \F_q^{\theta}C_{n} y$ we define $\Phi(a)=\sum_{i=0}^{n-1}a_i x^i \in \F_q^{\theta}C_{n}$.
\end{enumerate}
\end{definition}
Note that the map $\Phi:\F_q^{\theta}C_{n}y \rightarrow \F_q^{\theta}C_{n}$ is an $\F_q$-linear isomorphism.

\begin{lemma}\label{conm} Let $\theta: G=D_{2n}\rightarrow \Aut(\F_q)$ be a group homomorphism. If $\theta(x^i)(a)=a$ for all $i \in\{0, \ldots, n-1\}$ and for all $a\in \F_q$, then $a\hat{b}=b\hat{a}$ for $a,b\in \Gamma_\theta$.
\end{lemma} 
\begin{proof} Let $a=\sum_{i=0}^{n-1} a_i x^iy\in \Gamma_\theta$ and $b=\sum_{i=0}^{n-1} b_i x^iy \in \Gamma_\theta.$ Then 
$$a\hat{b}=\sum_{i=0}^{n-1} a_i x^iy \sum_{j=0}^{n-1} \theta (x^j y)(b_j) x^jy=\sum_{j=0}^{n-1} \left( \sum_{i=0}^{n-1}a_i \theta(x^j)(b_{[i-j]_n})\right) x^j$$ and
$$b\hat{a}=\sum_{i=0}^{n-1} b_i x^iy \sum_{j=0}^{n-1} \theta (x^j y)(a_j) x^jy=\sum_{j=0}^{n-1} \left( \sum_{i=0}^{n-1}b_i \theta(x^j)(a_{[i-j]_n})\right ) x^j. $$ Since $a, b\in \Gamma_\theta$, then $a_s=a_{[-s]_n}$ and $b_{[j-s]_n}=b_{[s-j]_n}$ for $s\in \{0, \ldots, n-1\}$. Therefore, $\theta(x^j)(a_{[-s]_n})=a_{[-s]_n}=a_s$ and $\theta(x^j)(b_{[s-j]_n})=b_{[s-j]_n}=b_{[j-s]_n}$, which is equivalent to the $s$-th term in $ \sum_{i=0}^{n-1}a_i \theta(x^j)(b_{[i-j]_n})$ coincides with the $[j-s]_n$-th term of $\sum_{i=0}^{n-1}b_i \theta(x^j)(a_{[i-j]_n})$.
\end{proof}

From now on, we will consider a square extension of $\F_q$, i.e., the skew group ring to consider will be $\F_{q^2}^{\theta}D_{2n}$. We have the following result whose proof follows straightforward.

\begin{lemma} \label{homo}Let $\sigma$ be a generator element of the Galois group $\Gal(\F_{q^2},\F_q)$. Then the map $\theta_\sigma: G=D_{2n} \rightarrow \Gal(\F_{q^2},\F_{q})$ defined by $\theta_\sigma(g)=\sigma$ for $g=x^iy$, $i\in\{0, \ldots, n-1\}$ and $\theta_\sigma(g)=1$ otherwise is a group homomorphism.
\end{lemma}
\begin{proof} 
This assertion can be checked straightforwardly.
\end{proof}

\section{A key exchange protocol}\label{KEP}

This section presents a key exchange protocol based on two-sided multiplications over a skew dihedral group ring. We remark that other works have considered two-sided semi-group actions or matrices over group rings for key exchange \cite{semi,matgroupRings,grouptheory,semigroupac}. However, we follow an alternative approach. Recently, in \cite{CV} the authors have proposed a key exchange protocol using two-sided multiplications over a dihedral twisted group ring. Following their construction, we introduce a similar key exchange protocol over $\F_{q^2}^{\theta}D_{2n}$, with $\theta$ being a suitable group homomorphism.

\subsection{The construction}\label{construction}
We start by setting up our key exchange protocol's public parameters.

\begin{enumerate}
 
 \item Choose $m,n\in \N$ and a prime number $p$ such that $p$ divides $n$. We then set $q=p^m$ and the finite field $\F_{q^2}$.
 \item Choose the map $\theta_\sigma: D_{2n} \rightarrow \Aut(\F_{q^2}) $ as it was defined in Lemma \ref{homo}. In particular, for $g=x^iy$ with $i\in\{0, \ldots, n-1\}$, $\theta_\sigma(g)=\sigma$, where $\sigma(a)=a^q$ for all $a\in \F_{q^2}$, and $\theta_\sigma(g)=1$ otherwise.
 
 \item Choose a random non-zero element $\texttt{h}_1 \in \F_{q^2}^{\theta} C_{n}$ and a random non-zero element $\texttt{h}_2 \in \F_{q^2}^{\theta} C_{n} y$. Set $\texttt{h}=\texttt{h}_1+\texttt{h}_2$ and make $\texttt{h}$ public.
\end{enumerate}

We use the notation introduced in~\cite{SecChannel}. Let $P_i$ and $P_j$ be two parties and $s$ be an identifier for a session. The key exchange protocol between $P_i$ and $P_j$ runs as shown by Protocol~\ref{protocol}.

\begin{algorithm}
\captionsetup{labelformat=empty}
\caption{\centering \textbf{Protocol 1} Our Key Exchange Protocol}
\label{protocol}
\begin{algorithmic}[1]

\State The initiator $P_i$, on input $(P_i, P_j, s)$, chooses a secret pair $(\texttt{a}_i,\gamma_i)\xleftarrow[]{R}\F_{q^2}^{\theta} C_{n} \times \Gamma_{\theta}$, and sends $(P_i, s,\texttt{pk}_{i}=\texttt{a}_i\texttt{h} \gamma_i)$ to $P_j$.
\State Upon receipt of $(P_i,s,\texttt{pk}_{i})$, $P_j$, chooses a secret pair $(\texttt{a}_j,\gamma_j)\xleftarrow[]{R}\F_{q^2}^{\theta} C_{n} \times \Gamma_{\theta}$ and sends $(P_j, s, \texttt{pk}_{j}=\texttt{a}_j \texttt{h} \gamma_j)$ to $P_i$, computes $\texttt{k}_j=\texttt{a}_j \texttt{pk}_{i} \widehat{\gamma_j}$, erases $(\texttt{a}_j,\gamma_j)$ and outputs the key $\texttt{k}_j$ under the session-id $s$.
\State Upon receipt of $(P_j,s,\texttt{pk}_{j})$, $P_i$, computes $\texttt{k}_i=\texttt{a}_i \texttt{pk}_{j} \widehat{\gamma_i}$, erases $(\texttt{a}_i,\gamma_i)$ and outputs the key $\texttt{k}_i$ under the session-id $s$.
\end{algorithmic}

\end{algorithm}

Note that if both $P_i$ and $P_j$ are uncorrupted during the exchange of the key and both complete the protocol for session-id $s$, then they both establish the same key. 
Because of the choice of $\theta_{\sigma}$, by Lemma \ref{conm} and Lemma \ref{homo} 
$$
 \texttt{k}_i=\texttt{a}_i \texttt{pk}_{j} \widehat{\gamma_i}= \texttt{a}_i \texttt{a}_j\texttt{h} \gamma_j \widehat{\gamma_i}=\texttt{a}_j \texttt{a}_i\texttt{h}\gamma_i \widehat{\gamma_j}=\texttt{a}_j \texttt{pk}_{i} \widehat{\gamma_j}=\texttt{k}_{j}
$$



\subsection{Intractability assumptions}
\label{InAssum}
With the notation above, let $\texttt{h}=\texttt{h}_1 +\texttt{h}_2$ be a public element in $\F_{q^2}^{\theta} D_{2n}$, where $\texttt{h}_1$ is a random non-zero element from $ \F_{q^2}^{\theta} C_{n}$ and $\texttt{h}_2$ is a random non-zero element from $\F_{q^2}^{\theta} C_{n} y$. We now present attack games~\cite{games2,games} for algebraic problems related to the security of our key exchange protocol. 


\begin{agame}[Skew Dihedral Product Decomposition]
\label{DPD}
 For a given adversary $\mathcal{A}$, we define the following attack game:
\begin{itemize}
 \item The challenger computes 
 \begin{algorithmic}[1]
 \State $(\texttt{a},\gamma)\xleftarrow[]{R}\F_{q^2}^{\theta} C_{n} \times \Gamma_{\theta}$;
 \State $\texttt{pk}\gets \texttt{a} \texttt{h}\gamma$;
 \end{algorithmic}
 and gives the value of $\texttt{pk}$ to the adversary.
 
 \item The adversary outputs $(\widetilde{\texttt{a}},\widetilde{\gamma}) \in \F_{q^2}^{\theta} C_{n} \times \Gamma_{\theta}.$

\end{itemize}

 We define $\mathcal{A}$'s advantage in solving the Skew Dihedral Product Decomposition Problem for $\F_{q^2}^{\theta} D_{2n}$, denoted $\textrm{SDPDadv}[\mathcal{A}, \F_{q^2}^{\theta} D_{2n}]$, as the probability that $\widetilde{\texttt{a}}\texttt{h}\widetilde{\gamma} =\texttt{a} \texttt{h}\gamma$. 

\end{agame}

\begin{definition}[Skew Dihedral Product Decomposition Assumption]We say that the Skew Dihedral Product Decomposition
(SDPD) assumption holds for $\F_{q^2}^{\theta} D_{2n}$ if for all efficient adversaries $\mathcal{A}$ the quantity $\textrm{SDPDadv}[\mathcal{A}, \F_{q^2}^{\theta} D_{2n}]$ is negligible.

\end{definition}

\begin{agame}[ Computational Skew Dihedral Product]
\label{CDP}
For a given adversary $\mathcal{A}$, we define the following attack game:
\begin{itemize}
 \item The challenger computes 
 \begin{algorithmic}[1]
 \State $(\texttt{a}_1,\gamma_1)\xleftarrow[]{R}\F_{q^2}^{\theta} C_{n} \times \Gamma_{\theta};$
 \State $(\texttt{a}_2,\gamma_2)\xleftarrow[]{R}\F_{q^2}^{\theta} C_{n} \times \Gamma_{\theta};$
 \State $\texttt{pk}_1\gets a_1 \texttt{h} \gamma_1;$
 \State $\texttt{pk}_2\gets a_2 \texttt{h} \gamma_2;$
 \State $\texttt{k}\gets \texttt{a}_2 \texttt{pk}_1 \widehat{\gamma}_2;$
 
 \end{algorithmic}
 and gives the values of $\texttt{pk}_1$ and $\texttt{pk}_2$ to the adversary.
 \item The adversary outputs some $\widetilde{\texttt{k}}\in \F_{q^2}^{\theta} D_{2n} $

\end{itemize}

 \noindent We define $\mathcal{A}$'s advantage in solving the Computational Skew Dihedral Product (CSDP) Problem for $\F_{q^2}^{\theta} D_{2n}$, denoted $\textrm{CSDPadv}[\mathcal{A}, \F_{q^2}^{\theta} D_{2n}]$, as the probability that $\widetilde{\texttt{k}}=\texttt{k}$.

\end{agame}

\begin{definition}[Computational Skew Dihedral Product Assumption]We say that the Computational Skew Dihedral Product
(CSDP) assumption holds for $\F_{q^2}^{\theta} D_{2n}$ if for all efficient adversaries $\mathcal{A}$ the quantity $\textrm{CSDPadv}[\mathcal{A}, \F_{q^2}^{\theta} D_{2n}]$ is negligible.
\end{definition}

\begin{lemma}If the SDPD assumption does not holds for $\F_{q^2}^{\theta} D_{2n}$, then CSDP assumption does not holds for $\F_{q^2}^{\theta} D_{2n}$.

\end{lemma}

\begin{proof} 
This assertion can be checked straightforwardly.
\end{proof}

\begin{agame}[Decisional Skew Dihedral Product]
\label{DSDP}
 For a given adversary $\mathcal{A}$, we define two experiments:

\vspace{2mm}

\noindent \textbf{Experiment} $\texttt{b}$

\begin{itemize}
 \item The challenger computes 
 
 \begin{algorithmic}[1]
 \State $(\texttt{a}_1,\gamma_1)\xleftarrow[]{R}\F_{q^2}^{\theta} C_{n} \times \Gamma_{\theta};$
 \State $(\texttt{a}_2,\gamma_2)\xleftarrow[]{R}\F_{q^2}^{\theta} C_{n} \times \Gamma_{\theta};$
 \State $(\texttt{a}_3,\gamma_3)\xleftarrow[]{R}\F_{q^2}^{\theta} C_{n} \times \Gamma_{\theta};$
 \State $\texttt{pk}_1 \gets \texttt{a}_1 \texttt{h} \gamma_1$; $\texttt{pk}_2 \gets \texttt{a}_2 \texttt{h} \gamma_2;$
 \State $\texttt{k}_0\gets \texttt{a}_2 \texttt{pk}_1 \widehat{\gamma}_2;~\texttt{k}_1 \gets \texttt{a}_3 \texttt{h} \gamma_3;$
 \end{algorithmic}
 and gives the triple $(\texttt{pk}_1,\texttt{pk}_2,\texttt{k}_{\texttt{b}})$ to the adversary.
 
 \item The adversary outputs a bit $\widetilde{\texttt{b}} \in \{0,1\} $

\end{itemize}

 \noindent Let $W_{\texttt{b}}$ is the event that $\mathcal{A}$ outputs $1$ in experiment $\texttt{b}$. We define $\mathcal{A}$'s advantage in solving the Decisional Skew Dihedral Product Problem for $\F_{q^2}^{\theta} D_{2n}$ as
 $$\textrm{DSDPadv}[\mathcal{A}, \F_{q^2}^{\theta} D_{2n}]=|\textrm{Pr}[W_0]-\textrm{Pr}[W_1]|.$$ 

\end{agame}

\begin{definition}[Decisional Skew Dihedral Product Assumption]We say that the Decisional Skew Dihedral Product
(DSDP) assumption holds for $\F_{q^2}^{\theta} D_{2n}$ if for all efficient adversaries $\mathcal{A}$ the quantity $\textrm{DSDPadv}[\mathcal{A},\F_{q^2}^{\theta} D_{2n}]$ is negligible.
\end{definition}
Note that $\texttt{h}$ is chosen as $\texttt{h}=\texttt{h}_1 +\texttt{h}_2$, with $\texttt{h}_1$ being a random non-zero element from $ \F_{q^2}^{\theta} C_{n}$ and $\texttt{h}_2$ being a random non-zero element from $\F_{q^2}^{\theta} C_{n} y$, to not let the attacker win the DSDP Game trivially. Indeed if $\texttt{h}$ is chosen as $\texttt{h}=\texttt{h}_1+\texttt{0}$ with $\texttt{h}_1 \in \F_{q^2}^{\theta} C_{n}$, then $\texttt{k}_0 \in \F_{q^2}^{\theta} C_{n}$ and $\texttt{k}_1 \in \F_{q^2}^{\theta} C_{n}y$ by Lemma~\ref{lemmap}. Similarly if $\texttt{h}$ is chosen as $\texttt{h}=\texttt{0}+\texttt{h}_2$ with $\texttt{h}_2 \in \F_{q^2}^{\theta} C_{n}y$, then $\texttt{k}_0 \in \F_{q^2}^{\theta} C_{n}y$ and $\texttt{k}_1 \in \F_{q^2}^{\theta} C_{n}$ by Lemma~\ref{lemmap}. Therefore the attacker can win the DSDP Game for both cases with non-negligible probability. Additionally, we have the following.
\begin{lemma}If the CSDP assumption
does not holds for $\F_{q^2}^{\theta} D_{2n}$, then DSDP assumption does not holds for $\F_{q^2}^{\theta} D_{2n}$.
\end{lemma}

\begin{proof} 
This assertion can be checked straightforwardly.
\end{proof}

\subsection{The Hardness of the SDPD Problem}

The authors of~\cite{CV} provide an algorithmic and algebraic analysis on the Dihedral Product Decomposition (DPD) Problem, which is the underlying problem associated with the security of their constructions. In particular, 
let $\F_q^{\alpha}D_{2n}$ be a twisted dihedral group algebra, where the 2-cocycle $\alpha_{\lambda}: D_{2n} \times D_{2n} \longrightarrow \F_q^* $ is defined by $\alpha_{\lambda}(g,h)=\lambda$ (a non-square in $\F_q$) for $g=x^iy$, $h=x^jy$ with $i,j \in\{0, \ldots, n-1\}$ and $\alpha_{\lambda}(g,h)=1$ otherwise. In~\cite{CV}, the authors demonstrate that $\F_q^{\alpha}D_{2n}=\F_q^{\alpha}C_{n}\oplus \F_q^{\alpha}C_{n}y$ as direct sum of $\F$-vector spaces, and also define $\Gamma_{\alpha_\lambda} \subseteq \F_q^{\alpha}C_{n}y$ in a similar way.
 
The DPD attack game is defined as follows. Let $\texttt{h}=\texttt{h}_1 +\texttt{h}_2$ be a public element in $\F_q^{\alpha_{\lambda}} D_{2n}$, where $\texttt{h}_1$ is a random non-zero element from $ \F_q^{\alpha} C_{n}$ and $\texttt{h}_2$ is a random non-zero element from $\F_q^{\alpha_{\lambda}} C_{n} y$. For a given adversary $\mathcal{A}$, 
\begin{itemize}
 \item The challenger computes 
 \begin{algorithmic}[1]
 \small
 \State $(\texttt{a},\gamma)\xleftarrow[]{R}\F_q^{\alpha_{\lambda}} C_{n} \times \Gamma_{\alpha_\lambda}$;
 \State $\texttt{pk}\gets \texttt{a} \texttt{h}\gamma$;
 \end{algorithmic}
 and gives the value of $\texttt{pk}$ to the adversary.
 
 \item The adversary outputs $(\widetilde{\texttt{a}},\widetilde{\gamma}) \in \F_q^{\alpha_{\lambda}} C_{n} \times \Gamma_{\alpha_\lambda}.$

\end{itemize}

The $\mathcal{A}$'s advantage in solving the Dihedral Product Decomposition Problem for $\F_q^{\alpha_{\lambda}} D_{2n}$ is defined as the probability that $\widetilde{\texttt{a}}\texttt{h}\widetilde{\gamma} =\texttt{a} \texttt{h}\gamma$. 

The authors of~\cite{CV} analyze how an adversary, with access to a quantum computer, may leverage it to try to solve the DPD problem by exploiting quantum algorithms (e.g., Grover's algorithm and Shor's algorithm)~\cite{quantCrypto}. Also, they analyze possible algebraic attacks on DPD problems and hence propose choosing their constructions' public parameters to avoid that an adversary may leverage those algebraic techniques, such as~\cite{twosided}, to solve the DPD problem. We remark that since the DPD problem and SDPD problem are very similar, such an algebraic and algorithmic analysis for the DPD problem presented in~\cite{CV} may be adapted easily to the SDPD problem. However, we note that adjusting such an analysis to the SDPD problem does not mean that both DPD and SDPD problems are computationally equivalent. It is indeed an open question to prove whether these problems are computationally equivalent or not.

\subsection{Security analysis in the authenticated-links adversarial model}

This subsection is devoted to analysing further our key exchange protocol in a appropriate security model~\cite{texchange1,texchange2,SecChannel}. In particular, we aim at proving that our protocol is session-key
secure in the authenticated-links adversarial model (AM) of Canetti and Krawczyk~\cite{SecChannel}, assuming the DSDP assumption holds for $\F_{q^2}^{\theta}$. We first recall the definition of session-key security in the authenticated-links adversarial model of Canetti and
Krawczyk~\cite{SecChannel}, and follow the description given in~\cite{CV}. 

\begin{enumerate}
    \item Let $P=\{P_1, P_2, \ldots, P_n\}$ be a finite set of parties.
    \item Let $\mathcal{A}$ be an adversary that controls all communication between two parties, however 
    \begin{itemize}
        \item  $\mathcal{A}$ is not allowed to inject or modify messages, except for messages sent by corrupted parties or sessions.
        
        \item $\mathcal{A}$ may choose not to forward a message at all, but if $\mathcal{A}$ chooses to forward a message $m$, $\mathcal{A}$ has to send it to the correct destination for $m$, only once and without modifying $m$.
        
        \item Parties give outgoing messages to $\mathcal{A}$, who has control over their delivery via the \texttt{Send} query. $\mathcal{A}$ can activate a party $P_i$ by \texttt{Send} queries, i.e. the adversary has control over the creation of protocol sessions, which take place within each party. Two sessions $s_1$ and $s_0$
are matching if the outgoing messages of one are the incoming messages of the other, and vice versa. Additionally, $\mathcal{A}$ is allowed to query the oracles \texttt{SessionStateReveal}, \texttt{SessionKeyReveal}, and \texttt{Corrupt}.
    
    \begin{itemize}
 \item If $\mathcal{A}$ query 
the \texttt{SessionStateReveal} oracle for a specified session-id $s$ within some party $P_i$, then $\mathcal{A}$ obtains the contents of the specified session-id $s$ within $P_i$, including any secret information. This event is noted and hence produces no further output. 

\item If $\mathcal{A}$ query the \texttt{SessionKeyReveal} for a specified session-id $s$, then $\mathcal{A}$ obtains the session key for the specified session $s$, assuming that $s$ has an associated session.

\item If $\mathcal{A}$ query the \texttt{Corrupt} oracle for a specified party $P_i$ , then $\mathcal{A}$ takes over the party $P_i$, i.e. $\mathcal{A}$ has access to all
information in $P_i$'s memory, including long-lived keys and any session-specific information still stored. A
corrupted party produces no further output.
\end{itemize}

    \item Finally $\mathcal{A}$ is given access to the \texttt{test} oracle, which can be queried once and at any stage to a completed, fresh,
unexpired session-id $s$. On input $s$, the \texttt{test} oracle chooses $b\xleftarrow[]{R}\{0,1\}$, then it outputs the session key for the specified session-id $s$ if $b=0$. Otherwise, it returns a random value in the key space. Also, $\mathcal{A}$ can issue subsequent queries as desired, with the exception that it cannot expose the test session. At any point, the adversary can try to guess $b$. Let $\textrm{Guess}[\mathcal{A}, \F_{q^2}^{\theta} D_{2n}]$ be the event that $\mathcal{A}$ correctly guesses $b$, and define the advantage $\textrm{SKAdv}[\mathcal{A}, \F_{q^2}^{\theta} D_{2n}]= |\textrm{Guess}[\mathcal{A}, \F_{q^2}^{\theta} D_{2n}]-1/2|$.

   \end{itemize}

\end{enumerate}

\begin{theorem} If the DSDP assumption
holds for $\F_{q^2}^{\theta} D_{2n}$, then our key exchange protocol is session-key secure in the the authenticated-links adversarial model, i.e. for any $\mathcal{A}$ in the authenticated-links adversarial model (AM), then the following holds 
\begin{enumerate}
 \item The key-exchange protocol satisfies the property that if two uncorrupted parties complete matching sessions, then they both output the same key.
 \item $\textrm{SKAdv}[\mathcal{A}, \F_{q^2}^{\theta} D_{2n}]$ is negligible.
\end{enumerate}

\end{theorem}

\begin{proof}
The proof of this theorem is an adaptation of the proof, given in~\cite{CV},  for the key-exchange protocol over a twisted dihedral group algebra $\F_q^{\alpha}D_{2n}$.

\begin{enumerate}
  
    \item The proof of the first statement is given at the end of the Subsection~\ref{construction}.



\item  To prove this statement, we proceed by contradiction. Let us suppose that there is an adversary $\mathcal{A}$ in the authentication-links model against our protocol that has a non-negligible advantage $\epsilon$ in guessing the bit $b$ chosen by the \texttt{test} oracle (when queried). Let $l$ be an upper bound on the number of sessions invoked by $\mathcal{A}$ in any interaction. We now present a distinguisher $\mathcal{D}$ for the DSDP problem.\\

\begin{algorithmic}[1]
\small
\Function{$\mathcal{D}$}{$h,\F_{q^2}^{\theta} D_{2n},\texttt{pk}_1,\texttt{pk}_2,\texttt{k}$}

\State $r\xleftarrow[]{R}\{1,\ldots, l\}$;
\State Invoke $\mathcal{A}$ on a simulated interaction in the AM with parties $P_1,\ldots, P_n$, except for the $r_{th}$ session;

\State For the $r$-th session, let $P_i$ send $(P_i, s, \texttt{pk}_{i}=\texttt{a}_i \texttt{h} \gamma_i)$ to $P_j$ and let $P_j$ send $(P_j, s, \texttt{pk}_{j}=\texttt{a}_j \texttt{h} \gamma_j)$ to $P_i$;

\If{the $r$-th session is selected by $\mathcal{A}$ as the test session}
\State Give $\texttt{k}$ to $\mathcal{A}$ as the answer to his query;
\State $d\gets \mathcal{A}(\texttt{k})$;
\Else

\State $d\xleftarrow[]{R}\{0,1\}$;
\EndIf
\State \Return $d$
\EndFunction

\end{algorithmic}

On the one hand, let us suppose that $\mathcal{A}$ picks the $r$-th as the test session, then $\mathcal{A}$ is provided with either $\texttt{k}_{\texttt{0}}$ or $\texttt{k}_{\texttt{1}}$, since the DSDP challenger gives either of the two keys to $\mathcal{D}$. Therefore, the probability that $\mathcal{A}$ correctly distinguishes is $1/2+\epsilon$ 
with non-negligible $\epsilon$ (by assumption). On the other hand,  assume that $\mathcal{A}$ does not choose the $r$-th as the test session, then $\mathcal{D}$ always returns a random bit, and hence the distinguishing probability for the input is $1/2$.

Note that the probability that the test session and the $r$-th session coincide is $1/l$. So these do not coincide with probability  $1-1/l$. Hence the overall probability for $\mathcal{D}$ to win the DSDP Game is $1/(2l)+\epsilon/l+1/2-1/(2l)=1/2+\epsilon/l$, which is non-negligible.
\end{enumerate}
\end{proof}

\section{Probabilistic Public Key Encryption}\label{V}

We now present a probabilistic public key encryption based on the key exchange protocol introduced in Section~\ref{KEP}. Following the notation above, choose a random non-zero element $\texttt{h}_1 \in \F_{q^2}^{\theta} C_{n}$ and a random non-zero element $\texttt{h}_2 \in \F_{q^2}^{\theta} C_{n} y$. Set $\texttt{h}=\texttt{h}_1+\texttt{h}_2$ and make $\texttt{h}$ public.


Let $\mathcal{SK}=\F_{q^2}^{\theta}C_n \times \Gamma_{\theta}$ be the secret key space, $\mathcal{PK}=\F_{q^2}^{\theta}D_{2n}$ be the public key space, $\mathcal{M}=\F_{q^2}^{\theta}D_{2n}$ be the message space, and $\mathcal{C}=\F_{q^2}^{\theta}D_{2n}$ the cipher-text space. We now define the public key encryption scheme $\mathcal{E}=(\texttt{Gen},\texttt{Enc}, \texttt{Dec})$.

\begin{algorithm}[H]

\caption{ Key Generation Algorithm}
\label{keygen}
\begin{algorithmic}[1]
\Function{\texttt{Gen}}{$\texttt{h} \in \F_{q^2}^{\theta} D_{2n}$}

\State $(\texttt{a}_1,\gamma_1) \xleftarrow[]{R} \mathcal{SK};$ 
\State $\texttt{pk} \gets \texttt{a}_1 \texttt{h} \gamma_1;$ 
\State $\texttt{sk} \gets (\texttt{a}_1,\gamma_1);$

\State \Return $\texttt{pk},\texttt{sk};$
\EndFunction
\end{algorithmic}
\end{algorithm}

\begin{algorithm}[H]

\caption{ Encryption Algorithm}
\label{encryption}
\begin{algorithmic}[1]
\Function{\texttt{Enc}}{$\texttt{m} \in \mathcal{M}, \texttt{pk} \in \mathcal{PK}, \texttt{r}_2 \in \mathcal{SK},\texttt{h} \in \F_{q^2}^{\theta} D_{2n}$}

\State $(\texttt{a}_2,\gamma_2)\gets \texttt{r}_2 ;$ 
\State $\texttt{c}_1 \gets \texttt{a}_2 \texttt{h} \gamma_2;$ 
\State $\texttt{c}_2 \gets \texttt{m}+ \texttt{a}_2 \texttt{pk} \widehat{\gamma}_2;$ 
\State $\texttt{c} \gets (\texttt{c}_1,\texttt{c}_2);$

\State \Return $\texttt{c};$
\EndFunction
\end{algorithmic}
\end{algorithm}

\begin{algorithm}[H]

\caption{ Decryption Algorithm}
\label{decryption}
\begin{algorithmic}[1]
\Function{\texttt{Dec}}{$\texttt{c} \in \mathcal{C},\texttt{sk} \in \mathcal{SK}$}
\State $(\texttt{a}_1,\gamma_1) \gets \texttt{sk};$
\State $(\texttt{c}_1,\texttt{c}_2) \gets \texttt{c};$
\State $\texttt{k} \gets \texttt{a}_1 \texttt{c}_1 \widehat{\gamma}_1 ;$ 
\State $\texttt{m} \gets \texttt{c}_2 - \texttt{k};$

\State \Return $\texttt{m};$
\EndFunction
\end{algorithmic}
\end{algorithm}

\begin{lemma}[Correctness]Let \texttt{h} be a public element in $\F_{q^2}^{\theta} D_{2n}$. Consider the encryption scheme $\mathcal{E}$ constructed above. For any message $\texttt{m} \in \mathcal{M}$, $\texttt{r}_{2}\xleftarrow[]{R} \mathcal{SK}$ and $(\texttt{pk}, \texttt{sk}) \gets \texttt{Gen}(\texttt{h}),$ it holds that $\texttt{m} \gets \texttt{Dec}(\texttt{Enc}(\texttt{m},\texttt{pk}, \texttt{r}_2, \texttt{h}),\texttt{sk})$

\end{lemma}

\begin{proof}Since $$(\texttt{c}_1=\texttt{a}_2 \texttt{h} \gamma_2,\texttt{c}_2=\texttt{m}+ \texttt{a}_2 \texttt{pk} \widehat{\gamma}_2) \gets \texttt{Enc}(\texttt{m}, \texttt{pk},\texttt{r}_2, \texttt{h}) $$ and $\texttt{sk}=(\texttt{a}_1,\gamma_1)$, then $$ \texttt{k}=\texttt{a}_1 \texttt{c}_1 \widehat{\gamma}_1=\texttt{a}_1 \texttt{a}_2 \texttt{h} \gamma_2 \widehat{\gamma}_1=\texttt{a}_2 \texttt{a}_1 \texttt{h} \gamma_1 \widehat{\gamma}_2 =\texttt{a}_2 \texttt{pk} \widehat{\gamma}_2, $$ and therefore $$
\texttt{c}_2 - \texttt{k}=\texttt{m}+ \texttt{a}_2 \texttt{pk} \widehat{\gamma}_2-\texttt{a}_2 \texttt{pk} \widehat{\gamma}_2=\texttt{m}$$
\end{proof}

\begin{theorem} If the DSDP assumption
holds for $\F_{q^2}^{\theta} D_{2n}$, then $\mathcal{E}$ is semantically secure.
\end{theorem}

\begin{proof} The proof of
Theorem $5.2$ in~\cite{CV} can be easily adapted to this setting.
\end{proof}

\section{A Key Encapsulation Mechanism}\label{VI}

By applying a generic transformation of Hofheinz, H\"ovelmanns, and Kiltz~\cite{Kem} to $\mathcal{E}$, we introduce a CCA-secure key encapsulation mechanism. Let $\mathcal{K}=\{0,1\}^{l_1}$ be the key space and $\texttt{rep}(x)$ be a function that simply returns the binary representation of $x$. Additionally, we construct the following two functions:

\begin{itemize}
\item $\mathcal{H}_1:\{0,1\}^{*} \longrightarrow \mathcal{SK}$ is a hash function that takes in a bit-string, say $\texttt{x}$, and then uses cryptographic hash function, e.g. $\texttt{SHAKE}_{256}$, to compute a key in the keyspace from it. Following the notation
of~\cite{SHA3}, $\mathcal{H}_1(\texttt{x})=\texttt{SHAKE}_{256}(\texttt{x}, \texttt{o})$, where $\texttt{o}=\lceil \log_2(p) \rceil 2m (n+\lceil\frac{n+1}{2}\rceil)$ is the bit length of the output. From this bit-string, the corresponding pair $(\texttt{a},\gamma) \in \mathcal{SK}$ can be obtained easily.

\item $\mathcal{H}_2:\{0,1\}^{*} \longrightarrow \mathcal{K}$ is a hash function that applies a cryptographic hash function, e.g. $\texttt{SHAKE}_{256}$, to the input. Specifically $\mathcal{H}_2(\texttt{x})=\texttt{SHAKE}_{256}(\texttt{p}_1 || \texttt{x}, l_1)$, where $\texttt{p}_1$ is a prepended fixed bit-string to make it different from $\mathcal{H}_1$. 

\end{itemize}

\noindent Applying the generic transformation ${\small \textrm{U}^{\not\perp}[\textrm{T}[\mathcal{E},\mathcal{H}_2], \mathcal{H}_1]}$ from~\cite{Kem}, we get $\textrm{KEM} = (\texttt{KeyGen},\texttt{Encaps},\texttt{Decaps})$.

\begin{algorithm}[H]

\caption{ Key Generation Algorithm}
\label{keygenE}
\begin{algorithmic}[1]
\Function{\texttt{KeyGen}}{$\texttt{h}$}

\State $(\texttt{pk} ,\texttt{sk}) \gets \texttt{Gen}(\texttt{h});$
\State $\texttt{s} \xleftarrow[]{R} \mathcal{M};$

\State \Return $(\texttt{s},\texttt{sk},\texttt{pk});$
\EndFunction
\end{algorithmic}
\end{algorithm}

\begin{algorithm}[H]

\caption{ Encapsulation Algorithm}
\label{foencryption}
\begin{algorithmic}[1]
\Function{\texttt{Encaps}}{$ \texttt{pk}, \texttt{h}$}

\State $\texttt{m} \xleftarrow[]{R} \mathcal{M};$
\State $\texttt{r} \gets \mathcal{H}_1(\texttt{rep}(\texttt{m})||\texttt{rep}(\texttt{pk}));$

\State $\texttt{c} \gets \texttt{Enc}(\texttt{m}, \texttt{pk} , \texttt{r},\texttt{h} );$
\State $\texttt{K} \gets \mathcal{H}_2(\texttt{rep}(\texttt{m}) || \texttt{rep}(\texttt{c}));$
\State \Return $(\texttt{c}, \texttt{K});$
\EndFunction
\end{algorithmic}
\end{algorithm}

\begin{algorithm}[H]

\caption{ Decapsulation Algorithm}
\label{fodecryption}
\begin{algorithmic}[1]
\Function{\texttt{Decaps}}{$(\texttt{s},\texttt{pk},\texttt{sk}), \texttt{c},\texttt{h}$}
\State $\texttt{m} \gets \texttt{Dec}(\texttt{c},\texttt{sk} );$
\State $\texttt{r} \gets \mathcal{H}_1(\texttt{rep}(\texttt{m})||\texttt{rep}(\texttt{pk}));$
\If {$\texttt{c} = \texttt{Enc}(\texttt{m}, \texttt{pk} , \texttt{r},\texttt{h} )$}
\State $\texttt{K} \gets \mathcal{H}_2(\texttt{rep}(\texttt{m}) || \texttt{rep}(\texttt{c}));$
\State \Return $\texttt{K};$
\Else 
\State \Return $\mathcal{H}_2(\texttt{rep}(\texttt{s})||\texttt{rep}(\texttt{c}));$
\EndIf

\EndFunction
\end{algorithmic}
\end{algorithm}

\section{Implementation of our cryptographic constructions}\label{implementations}

We implemented our proposed public-key encryption scheme and key encapsulation mechanism as a proof-of-concept in Python. The interested reader can see it on Google Colaboratory~\cite{Impl}. 

\subsection{Dihedral Group }

To implement a dihedral group of order $2n$, we simply represent a dihedral group element $g=x^{i_1}y^{j_1}$ as the integer $k_1=j_1\cdot n+i_1$. Also, we compute a $2n\times 2n$ integer array \texttt{table} such that the row $\texttt{table}[k_1]$, $0\leq k_1<2n$, stores a $2n$ array with the integer representations of $g,gx,gx^2,\ldots gx^{n-1},gy,\ldots, gx^{n-1}y $. To compute the operation of two given group elements $g=x^{i_1}y^{j_1}$ and $h=x^{i_2}y^{j_2}$, we simply return $\texttt{table}[k_1][k_2] $, where $k_1=j_1\cdot n+i_1$ and $k_2=j_2\cdot n+i_2$. To compute the multiplicative inverse of a given group element $g=x^{i_1}y^{j_1}$, the function $\texttt{inverse}(k_1) $ returns $0$ if $k_1=0$, or $n-k_1$ if $1\leq k_1<n$, or $k_1$ if $n\leq k_1<2n$. 

\subsection{Homomorphism $\theta$ }

The homomorphism $\theta$ is implemented as described next. Given $k_1$ and $k_2$, two representations of two group elements, then the function $\texttt{homomorphism}(k_1,k_2)$ returns a pointer to the function $\sigma$ if $n \leq k_1<2n$~\textrm{and}~$n \leq k_2<2n$. Otherwise, it returns a pointer to the function identity $\texttt{I}$. We will next describe the inner working of each function.

\begin{algorithmic}[1]
\small
\Function{$\sigma$}{$a \in \F_{q^2}^{\theta} $}
\State $[b_s,b_{s-1},\ldots ,b_{0}] \gets \texttt{getBinaryReprepresation}(q);$
\State $r\gets \texttt{getOneFromQuadraticField}();$

\For{$i\gets s ~to~0$}
\State $r\gets r\cdot r$
 \If{$b_i=1$}
 \State $r\gets r\cdot a$
 \EndIf
\EndFor

\State \Return $r$ 
\EndFunction
\end{algorithmic}

\bigskip

\begin{algorithmic}[1]
\small
\Function{$\texttt{I}$}{$a \in \F_{q^2}^{\theta} $}
 \State \Return $a$
\EndFunction
\end{algorithmic}

\subsubsection{The skew dihedral group ring $\F_{q^2}^{\theta}D_{2n}$ }
An element $a=\sum_{i=0}^{n-1}a_i x^i+\sum_{i=0}^{n-1}a_{n+i} x^iy$ in the group ring $\F_{q^2}^{\theta}D_{2n}$ is represented as a array of $2n$ field elements $\texttt{a}=[\texttt{a}_0,\texttt{a}_1, \texttt{a}_2,\ldots, \texttt{a}_{2n-1} ]$, where $\texttt{a}_i$ is the representation of the field element $a_i \in \F_q^2$. Therefore the addition and product of two elements of this ring is easily implemented as shown next.

\begin{multicols}{2}
\begin{algorithmic}[1]
\small
\Function{\texttt{addition}}{$\texttt{a},\texttt{b}$}
\State $\texttt{c} \gets [\texttt{0},\cdots, \texttt{0}]$ 

 \For {$(i\gets 0;i<2n;i\gets i+1)$}
 \State $\texttt{c}[i]\gets \texttt{a}[i]+\texttt{b}[i]$;
 
 \EndFor
\State \Return \texttt{c}
\EndFunction
\end{algorithmic}

\columnbreak

\begin{algorithmic}[1]
\small
\Function{\texttt{product}}{$\texttt{a},\texttt{b}$}

\State $\texttt{c} \gets [\texttt{0},\cdots, \texttt{0}]$

 \For {$(i\gets 0;i<2n;i\gets i+1)$}
 
 \For {$(j\gets 0;j<2n;j\gets j+1)$}
 \State $k\gets\texttt{table}[i,j]$;
 \State $\texttt{fe}\gets \texttt{a}[i]\cdot(\texttt{homomorphism}(i)(b[j]))$;
 \State $\texttt{c}[k]\gets \texttt{c}[k]+\texttt{fe}$;
 \EndFor
 \EndFor
\State \Return \texttt{c}
\EndFunction
\end{algorithmic}
\end{multicols}

On the one hand, the addition function has a cost of $2n$ field additions to compute an ring element $\texttt{c}$. On the other hand, the product function has a cost of $4n^2$ field additions and $4n^2\cdot (1+f)$ field multiplications, where $f$ is the number of field multiplication to compute $\texttt{homomorphism}(i)(b[j])$. In addition to these functions, we implement the function \texttt{adjunct}, which computes the adjunct of a ring element and its cost is $2n\cdot f$ multiplications. Also, functions for computing a random element in $\Gamma_{\theta}$ ( $\F_{q^2}^{\theta}D_{2n}$, $\F_{q^2}^{\theta}C_{n}$ and $\F_{q^2}^{\theta}C_{n}y$) are described next. 

\begin{multicols}{2}
\begin{algorithmic}[1]
\small
\Function{\texttt{adjunct}}{$\texttt{a}$}
\State $\texttt{c} \gets [\texttt{0},\cdots, \texttt{0}]$

 \For {$(i\gets 0;i<2n;i\gets i+1)$}
 \State $j\gets \texttt{inverse}(i)$
 \State $\texttt{c}[j] \gets \texttt{homomorphism}(j)(a[i])$
 
 \EndFor
\State \Return \texttt{c}
\EndFunction
\end{algorithmic}

\begin{algorithmic}[1]
\small
\Function{\texttt{$\texttt{getRandomfromT}$}}{$ $}
\State $\texttt{c} \gets [\texttt{0},\cdots, \texttt{0}]$
 \State $\texttt{c}[n] \gets \texttt{getRandomFieldElement}()$
 \State $n_1 \gets n/2$
 \For {$(i\gets 1;i\leq n_1;i\gets i+1)$}
 \State $\texttt{c}[i+n]\gets \texttt{getRandomFieldElement}()$
 \State $\texttt{c}[n+(n-i)~\textrm{mod}~n]\gets \texttt{c}[i+n]$

 \EndFor
\State \Return \texttt{c}
\EndFunction
\end{algorithmic}

\columnbreak

\begin{algorithmic}[1]
\small
\Function{\texttt{$\texttt{getRandomFD2n}$}}{$ $}
\State $\texttt{c} \gets [\texttt{0},\cdots, \texttt{0}]$
 
 \For {$(i\gets 0;i< 2n;i\gets i+1)$}
 \State $\texttt{c}[i]\gets \texttt{getRandomFieldElement}()$

 \EndFor
\State \Return \texttt{c}
\EndFunction
\end{algorithmic}

\begin{algorithmic}[1]
\small
\Function{\texttt{$\texttt{getRandomFCn}$}}{$ $}
\State $\texttt{c} \gets [\texttt{0},\cdots, \texttt{0}]$
 
 \For {$(i\gets 0;i< n;i\gets i+1)$}
 \State $\texttt{c}[i]\gets \texttt{getRandomFieldElement}()$

 \EndFor
\State \Return \texttt{c}
\EndFunction
\end{algorithmic}

\begin{algorithmic}[1]
\small
\Function{\texttt{$\texttt{getRandomFCny}$}}{$ $}
\State $\texttt{c} \gets [\texttt{0},\cdots, \texttt{0}]$
 
 \For {$(i\gets n;i< 2n;i\gets i+1)$}
 \State $\texttt{c}[i]\gets \texttt{getRandomFieldElement}()$

 \EndFor
\State \Return \texttt{c}
\EndFunction
\end{algorithmic}

\end{multicols}

We also implement the following function to compute a random public element $h$.

\bigskip

\begin{algorithmic}[1]
\small
\Function{\texttt{$\texttt{getPublicElement}$}}{$ $}
\State $\texttt{sw}_1 \gets \textbf{False}$
\While{$\texttt{not}~\texttt{sw}_1$}
\State $\texttt{a}\gets \texttt{getRandomFD2n}()$
\State $i\gets 0$
\State $\texttt{sw}_2 \gets \textbf{False}$
\While{$i<n~\texttt{and}~\texttt{not}~\texttt{sw}_2$}

\If{$\texttt{a}[i]\neq\texttt{0}$} 
\State $\texttt{sw}_2 \gets \textbf{True}$
\EndIf
\State $i\gets i+1$
\EndWhile

\State $i\gets n$
\State $\texttt{sw}_3 \gets \textbf{False}$
\While{$i<2n~\texttt{and}~\texttt{not}~\texttt{sw}_3$}

\If{$\texttt{a}[i]\neq \texttt{0}$} 
\State $\texttt{sw}_3 \gets \textbf{True}$
\EndIf
\State $i\gets i+1$
\EndWhile
\State $\texttt{sw}_1 \gets \texttt{sw}_2~ \texttt{and}~\texttt{sw}_3$

\EndWhile

\State \Return \texttt{a}
\EndFunction
\end{algorithmic}

\subsubsection{Parameters choice}
For our KEM, we propose to use the parameters shown by Table~\ref{tab:sec}, which provide varying degrees of security.

\begin{table}[H]
\centering
\begin{tabular}{c|c|c|c|c}
\hline
$p$ & $m$ & $n$ & $l_1$ (bits) & level of security in bits \\ \hline
$19$ & $1$ & $19$ & $\{128, 192, 256\}$ & $124$ \\ \hline
$23$ & $1$ & $23$ & $\{128, 192, 256\}$ & $149$ \\ \hline
$31$ & $1$ & $31$ & $\{128, 192, 256\}$ & $200$ \\ \hline
$41$ & $1$ & $41$ & $\{128, 192, 256\}$ & $264$ \\ 
\end{tabular}
\caption{Proposed parameters }
\label{tab:sec}
\end{table}

Table~\ref{tab:sec} shows four sets of parameters providing various degrees of security, where $l_1 \in \{128,192,256\}$ refers to the length of the output key. The values in the level of security column were calculated as proposed in~\cite{CV}. To see the code of our implementation, please see~\cite{Impl}.

\end{document}